\documentclass[aps,pra,twocolumn,superscriptaddress,floatfix,
nofootinbib,showpacs,longbibliography]{revtex4-1}

\usepackage[utf8]{inputenc}  
\usepackage[T1]{fontenc}     
\usepackage[british]{babel}  
\usepackage[sc,osf]{mathpazo}
\usepackage[scaled=0.86]{berasans}  
\usepackage[colorlinks=true, citecolor=blue, urlcolor=blue]{hyperref}  
\usepackage{graphicx} 
\usepackage[babel]{microtype}  
\usepackage{amsmath,amssymb,amsthm,bm,amsfonts,mathrsfs,bbm} 

\usepackage{xspace}  
\usepackage{pgf,tikz}
\usepackage{xcolor}
\usepackage{multirow}
\usepackage{array}
\usepackage{bigstrut}
\usepackage{braket}
\usepackage{color}
\usepackage{natbib}
\usepackage{multirow}
\usepackage{float}
\usepackage[caption = false]{subfig}
\usepackage{xcolor,colortbl}
\usepackage{color}
\newcommand{\Tr}{\operatorname{Tr}}

\newcommand{\be}{\begin{equation}}
\newcommand{\ee}{\end{equation}}
\newcommand{\ba}{\begin{eqnarray}}
\newcommand{\ea}{\end{eqnarray}}

\newcommand{\tr}{\operatorname{Tr}}

\newtheorem{theorem}{Theorem}

\newtheorem{definition}{Definition}

\newtheorem{observation}{Observation}

\newtheorem{remark}{Remark}

\def\>{\rangle}
\def\<{\langle}

\makeatletter
\DeclareRobustCommand{\cev}[1]{%
	{\mathpalette\do@cev{#1}}%
}
\newcommand{\do@cev}[2]{%
	\vbox{\offinterlineskip
		\sbox\z@{$\m@th#1 x$}%
		\ialign{##\cr
			\hidewidth\reflectbox{$\m@th#1\vec{}\mkern4mu$}\hidewidth\cr
			\noalign{\kern-\ht\z@}
			$\m@th#1#2$\cr
		}%
	}%
}
\makeatother

\begin{document}
\title{The principle of information symmetry constrains the state-space in any physical theory}

\author{Manik Banik}
\affiliation{S.N. Bose National Center for Basic Sciences, Block JD, Sector III, Salt Lake, Kolkata 700098, India.}

\author{Sutapa Saha}
\affiliation{Physics and Applied Mathematics Unit, Indian Statistical Institute, 203 B.T. Road, Kolkata-700108, India.}

\author{Tamal Guha}
\affiliation{Physics and Applied Mathematics Unit, Indian Statistical Institute, 203 B.T. Road, Kolkata-700108, India.}

\author{Sristy Agrawal}
\affiliation{Department of Physical Sciences, Indian Institute of Science Education and Research Kolkata,\\ Mohanpur 741246, West Bengal, India.}

\author{Some Sankar Bhattacharya}
\affiliation{Department of Computer Science, The University of Hong Kong, Pokfulam Road, Hong Kong.} 

\author{Arup Roy}
\affiliation{S.N. Bose National Center for Basic Sciences, Block JD, Sector III, Salt Lake, Kolkata 700098, India.}

\author{A. S. Majumdar}
\affiliation{S.N. Bose National Center for Basic Sciences, Block JD, Sector III, Salt Lake, Kolkata 700098, India.}

\begin{abstract}
Symmetry shares an entwined history with  the structure of physical theory. We propose a consequence of symmetry towards the axiomatic derivation of Hilbert space quantum theory.
We introduce the notion of information symmetry (IS) and show that it constraints the state-space structure in any physical theory. To this end we study the minimal error binary state discrimination problem in the framework of generalized probabilistic theories. A theory is said to satisfy IS if the probability of incorrectly identifying each of two randomly prepared states is same for both the states. It is found that this simple principle rules out several classes of theories while being perfectly compatible with quantum theory.
 \end{abstract}



\maketitle
\emph{Introduction.--} 
Obtaining a physical perspective of the abstract mathematical description of quantum theory is a long-standing aspiration in quantum foundations. A variety of different approaches, some as old as the theory itself, have attempted to addressed this question, providing deeper understanding about the Hilbert space formulation of the theory  \cite{Birkhoff36,Beltrametti81,Soler95,Gleason57,Haag64,Haag96,Mackey'63,Ludwig'67,Mielnik'68,Clifton03,Abramsky04}. The advent of quantum information theory introduces a new direction to this endeavour. It identifies physically motivated principles excluding a class of multipartite {\it nonlocal} correlations that are strong enough to be incompatible with quantum theory, though weak enough to satisfy relativistic causality or the no-signalling (NS) principle \cite{vanDam, Brassard06, Linden07, Pawlowski09, Navascues09, Fritz13, Das13, Kunkri17, Bhattacharya17, Aravinda18}, thus providing a device-independent outlook about the correlations al
lowed in the physical world \cite{Scarani13,Brunner14}. Another approach is to identify rudimentary rule(s) that directly derive the state space structure or some crucial features of quantum theory \cite{Hardy'01,Aaronson04,Barrett'07,Barnum10,Acin10,Oppenheim10,Masanes'11,Chiribella'11,Dakic11,Muller12,Pfister13,Banik13,Cabello13,Banik15,Chiribella16,Czekaj18,Krumm19,Cabello19}.

Despite a number of non-trivial achievements, a complete physical or first-principles motivation of Hilbert space quantum mechanics is still elusive. In the present work we consider a different approach to address this issue, by investigating the state space structure of  physical theories from the perspective of {\it symmetry, as a principle}. Symmetry has played a long and widespread role in formulating theories of the physical world. Rather than being the by-product of dynamical laws, symmetry principles have been appreciated as primary features of nature, that in turn, determine the fundamental physical laws \cite{Watanabe55,Gross96}. For instance, while formulating the special theory of relativity, Einstein recognized relativistic invariance as a principle, which stipulates the form of transformation rules to be Lorentzian. Later, a similar approach guided him to develop his seminal theory of gravity where the principle of equivalence -- a principle of l
ocal symmetry -- determines the dynamics of space-time. In the context of the present work too, we take symmetry as the guiding feature, though the symmetry we explore here has different consequences. Rather than guiding directly the dynamics, it imposes constraints on the ways of information gain in the act of measurement, and consequently puts restrictions on the structure of state space. 

In order to study the implications of the proposed symmetry, we consider a very generic mathematical framework that allows the largest possible class of convex operational theories, also called generalized probability theories (GPTs). The state space of such a theory is a convex set in $\mathbb{R}^n$ with extreme points \cite{Supple} denoting  pure states or states of maximal knowledge. This framework embraces the notion of indistinguishable states -- members of a set of states that can not be identified perfectly given a single copy of the system prepared in one of these states. For a completely random ensemble of two such states, the most general strategy for minimum-error discrimination comprises of a two-outcome measurement -- the two different outcomes correspond to two different preparations. While extracting information through such a binary measurement,  error can occur in two ways:- (i) outcome-1 that should correspond to state-1 may click even when the system is 
prepared in state-2,  and (ii) outcome-2 may click when the system is prepared in state-1. Our proposed {\it Information Symmetry} (IS) assumes that for any randomly prepared binary ensemble of pure states, optimal information about the preparation is obtained symmetrically from both the states. In other words, the two possible sources of error contribute equally in minimal error state discrimination. Throughout the paper we consider that the pair of states are prepared with uniform probability distribution.

Through the analysis presented herein, we find that this seemingly naive symmetry condition is not satisfied by a large class of GPTs. In particular, we show that regular polygonal state spaces \cite{Janotta'11} with more than $4$ pure states are incompatible with IS. Polygonal state spaces with $4$ pure states, known by the name {\it squit}, also become incompatible with IS when it is applied to the binary ensembles of mixed states. This newly identified symmetry property turns out to be pivotal in determining the state space structure of physical theories as we find that both classical and quantum theory are perfectly compatible with IS. We begin our analysis with a brief discussion on the mathematical framework of GPTs.

\emph{Framework.--} The structure of any operational theory consists of three basic notions -- state or preparation, observable or measurement, and transformation \cite{Hardy'01,Barrett'07,Masanes'11,Chiribella'11}. While observables correspond to the possible choices of measurement on the system, its initial preparation is represented by a state, and the time evolution of the state is governed  by some transformation rule. In the {\it prepare and measure} scenario the state and observable together yields the statistical prediction of an outcome event. 

Preparation or state $\omega$ of a system specifies outcome probabilities for all measurements that can be performed on it. A complete specification of the state is achieved by listing the outcome probabilities for measurements belonging to a `fiducial set’ \cite{Hardy'01,Barrett'07}. The set $\Omega$ of all possible states is a compact and convex set embedded in the positive convex cone $V_+$ (see \cite{Supple} for precise definition) of some real vector space $V$. Convexity of $\Omega$ assures that any statistical mixture of states is a valid state. The extremal points of the set $\Omega$ are called pure states. For example, state of a quantum system associated with Hilbert space $\mathcal{H}$ is described by positive semi-definite operator with unit trace, {\it i.e.}, a density operator $\rho\in\mathcal{D}(\mathcal{H})$, where $\mathcal{D}(\mathcal{H})$ denotes the set of density operators acting on $\mathcal{H}$. For the simplest two level quantum syste
m (also called a qubit) $\mathcal{D}(\mathbb{C}^2)$ is isomorphic to a unit sphere in $\mathbb{R}^3$ centered at the origin, where points on the surface correspond to pure states.

An effect $\mathit{e}$ is a linear functional on $\Omega$ that maps each state onto a probability $p(\mathit{e}|\omega)$ representing successful filter of the effect $e$ on the state $\omega$. Unit effect $u$ is defined as, $p(u|\omega)=1,~\forall~\omega\in\Omega$. The set of all linear functionals forms a convex set embedded in the cone $V_+^*$ dual to the state cone $V_+$. The set of effects is occasionally denoted as $\Omega^*\subset V_+^*$. A $d$-outcome measurement $M$ is specified by a collection of $d$ effects, {\it i.e.}, $M\equiv\{\mathit{e}_j~|~\sum_je_j=u\}$. For every effect $e$ one can always construct a dichotomic measurement $M:=\{e,\bar{e}\}$ such that $p(e|\omega)+p(\bar{e}|\omega)=1,~\forall~\omega\in\Omega$; $\bar{e}$ is called the complementary effect of $e$. Likewise the states, effects can also be characterized as pure and mixed ones. Framework of GPTs may assume, a priori, that not all mathematically well-defined states are allowed physi
cal states and not all mathematically well-defined observables are allowed physical operations. For example, the set of physically allowed effects $\mathcal{E}$ may be a strict subset of $\Omega^*$. A theory is called ‘dual’ if it allows all elements of $\Omega^*$ as valid effects \cite{Self1}. In this generic framework of probabilistic theory, one can define the notion of distinguishable states. 
\begin{definition}
	Members of a set of $n$ states $\{\omega_i\}_{i=1}^n\subset\Omega$ are called distinguishable if they can be perfectly identified in a single shot measurement {\it i.e.}, if there exists an $n$-outcome measurement $M=\{e_j~|~\sum_{j=1}^ne_j=u\}$ such that $p(\mathit{e}_j|w_i)=\delta_{ij}$.
\end{definition}
Not every set of states can be perfectly discriminated. However, a set of such indistinguishable states can be distinguished probabilistically allowing one to define the following  state discrimination task. Suppose one of the states chosen randomly from the pair $\{\omega_1, \omega_2\}\subset\Omega$ is given. The aim is to optimally guess the correct state while one copy of the system is provided. Without loss of generality one can perform a two outcome measurement $M=\{e_1,e_2~|~e_1+e_2=u\}$ and guess the state as $\omega_i$ while the effect $e_i$ clicks. The error in guessing can occur in two ways -- effect $e_1$ clicks when the given state is actually $\omega_2$ which happen with probability $p_{12}:=p(e_1|\omega_2)$, and with $p_{21}:=p(e_2|\omega_1)$ probability effect $e_2$ clicks when the given state is actually $\omega_1$. As the states are chosen with uniform probability, the total error is therefore $p_E=\frac{1}{2}(p_{12}+p_{21})$, and hence, the 
probability of successful guessing is  $p_I=1-p_E$. The measurement that minimizes the error $p^{\min}_E:=\min_Mp_E$ is known as the Helstrom measurement, initially studied for quantum ensembles in 1970's \cite{Helstrom'69,Holevo73,Yuen75} and more recently, also studied in the GPT framework \cite{Kimura'09,Bae11,Nuida'10,Bae16}. While $e_1$ and $e_2$ used in the above discrimination task are mixed effects in general, however in the Helstrom measurement one of them is a pure effect.
\begin{remark}\label{remark1}
	For any pair of indistinguishable states in a GPT the measurement that optimally discriminates the states consists of a pure effect and its complementary effect. 
\end{remark}
In a GPT a pure state corresponds to the state of maximal knowledge. While in binary state-discrimination problem a pair of such states are given randomly with uniform probability distribution, it seems that both states should contribute identically in the error probability of optimal guessing. This leads us to the following definition.  
\begin{definition}
	A GPT is said to satisfy information symmetry (IS) if $p_{12}=p_{21}$ in $p_E^{\min}$ for every pair of pure states allowed in that GPT. In orther words, for any pair of pure
	states, maximum information about the ensemble is obtained only if both states contribute symmetrically to this quantity.
\end{definition}
Classical theory trivially satisfies IS as all the pure states are perfectly distinguishable. The classical state space with $d$ number of perfectly distinguishable states is a $(d-1)$-simplex. In quantum mechanics,  there however exists indistinguishable pure states. For a pair of such pure states, $\psi\equiv\ket{\psi}\bra{\psi},~\phi\equiv\ket{\phi}\bra{\phi}\in\mathcal{D}(\mathcal{H})$ the minimum error state discrimination (MESD) is obtained through Helstrom measurement \cite{Helstrom'69,Holevo73,Yuen75}. While $\psi$ and $\phi$ are prepared randomly with equal probability, the measurement $M\equiv\left\lbrace E_{\psi},~E_{\phi}~|~E_{\psi},~E_{\phi}\in\mathcal{L}^+(\mathcal{H})~s.t.~E_{\psi}+E_{\phi}=\mathbb{I}\right\rbrace $ achieving MESD is the one consisting of projectors onto the basis that ``straddles" $\psi$ and $\phi$ in Hilbert space, and we have $p_E^{\min}=\frac{1}{2}\left( 1-\sqrt{1-|\braket{\psi|\phi}|^2}\right)$ \cite{Supple}; $\mathcal{L}^+(\mathcal{H})$ i
s the set of positive operators on $\mathcal{H}$. Although IS holds true in classical and quantum theory, we  now show that the class of GPTs with regular polygonal state spaces are not compatible with it. 
\begin{figure}[t]
	\includegraphics[scale=0.3]{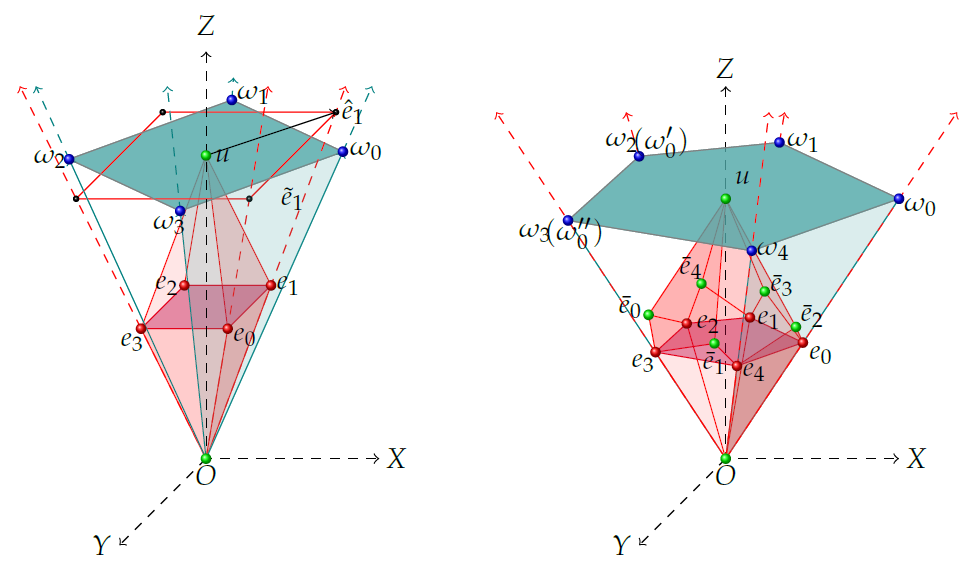}
	\caption{(Color on-line) State and effect spaces of {\it squit} (left) and {\it pentagon} (right) models. Blue dots are the extremal states and red dots denote the ray extremal effects. Green dots denote extremal effects that are not ray extremal.  In the squit model the effect $e_1$ is scaled up to $\tilde{e}_1$ so that its tip (black dot) lies on the normalized states space (green surface) and it can be represented as $\tilde{e}_1=u+\hat{e}_1$. In the pentagon model, the state $\omega_{0}$ and the states $\omega_{0}^{(\eta)}:=\eta\omega_{0}^{\prime}+(1-\eta)\omega_{0}^{\prime\prime}\equiv\eta\omega_2+(1-\eta)\omega_3$ are perfectly distinguishable by the dichotomic measurement $M\equiv=\{e_0,\bar{e}_0\}$, for all $\eta\in[0,1]$.}\label{poly}
\end{figure}

\emph{Regular polygonal state spaces.--} An associated toy theory for bipartite systems was first proposed to demonstrate the possibility of no-signaling theories which can have nonlocal behavior similar to quantum mechanics \cite{Popescu'94}. This entails the need to exclude such theories by providing new {\it physical} principles. In fact, several successful attempts have been made to exclude stronger than quantum nonlocal correlations \cite{vanDam, Brassard06, Linden07, Pawlowski09, Navascues09, Fritz13, Das13, Kunkri17, Bhattacharya17, Aravinda18}. Here we take a different approach. We aim to exclude a large class of such theories by invoking principle(s) that consider only the elementary system, {\it i.e.}, single partite system.

For an elementary system the state space $\Omega_n$ is a regular polygon with $n$ vertices \cite{Janotta'11,Weis12,Janotta13,Massar14,Janotta14,Safi15,Bhattacharya18}. For a fixed $n$, $\Omega_n$ is the convex hull of $n$ pure states $\{\omega_i\}_{i=0}^{n-1}$ with $\omega_i:=\left(r_n \cos(\frac{2 \pi i}{n}), r_n \sin(\frac{2 \pi i}{n}),1 \right)^T\in\mathbb{R}^3$; where $T$ denotes transpose and $r_n:= \sqrt{\sec(\pi/n)}$. The unit effect is given by $u:=(0,0,1)^T$. The set $\mathcal{E}$ of all possible measurement effects consists of convex hull of zero effect, unit effect, and the extremal effects $\{e_i,\bar{e}_i\}_{i=0}^{n-1}$, where $e_i:= \frac{1}{2}\left(r_n \cos(\frac{(2 i-1) \pi}{n}),r_n \sin(\frac{(2 i-1) \pi}{n}),1\right)^T$ for even $n$ and $e_i:=\frac{1}{1 + {r_n}^2}\left(r_n \cos(\frac{2 \pi i}{n}),r_n \sin(\frac{2 \pi i}{n}),1\right)^T$ for odd $n$.

The pure effects $\{e_i\}_{i=0}^{n-1}$ correspond to exposed rays and consequently the extreme rays of $V_+^*$ \cite{Supple,Yopp07}. For odd-gonal cases, due to self-duality of state cone $V_+$ and its effect cone $V_+^*$ \cite{Kimura14} every pure effect $e_i$ has one to one ray-correspondence to the pure state $\omega_i$. Consequently, for every pure state $\omega_i$ there exist exactly two other pure states $\omega_i^\prime$ and $\omega_i^{\prime\prime}$ such that $\omega_i$ and $\bar{\omega}_i^{(\eta)}:=\eta\omega_i^\prime+(1-\eta)\omega_i^{\prime\prime}$ are always perfectly distinguishable for all $\eta\in[0,1]$ (see Fig. \ref{poly}). The discriminating measurement consists of the effects $\{e_i,\bar{e}_i\}$ such that $p(e_i|\omega_i)=1$ and $p(e_i|\bar{\omega}_i^{(\eta)})=0$. The effects $\{\bar{e}_i\}_{i=0}^{n-1}$ are extremal elements of $\mathcal{E}$ but they are not ray extremal, i.e., they do not lie on an extremal ray of the cone $V_+^*$ \cite{Self2}. For an even
-gon, t
he scenario is quite different as the self duality between $V_+$ and $V_+^*$ is absent. Here all the $e_i$'s and their complementary effects $\bar{e}_i$'s correspond to extreme rays of $V_+^*$.

Every ray-extremal effect $e$ generates an extreme ray $\lambda e$ for the cone $V_+^*$, where $\lambda\ge 0$. With proper choice of $\lambda$ any such $e$ can be scaled up to a new $\tilde{e}\equiv\lambda e$, such that the tip of this scaled effect vector $\tilde{e}$ lies on the normalized state plane. Let us consider a particular direct sum decomposition of the space $\mathbb{R}^3$, {\it i.e.}, $\mathbb{R}^3=\mathbb{R}u\oplus \hat{V}$, where $\hat{V}$ is a two dimensional subspace of $\mathbb{R}^3$ parallel to the $X-Y$ plane. This allows a particular representation of $\tilde{e}$ in the following way $\tilde{e}=u+\hat{e}$, where $\hat{e}\in\hat{V}$. Similarly, every $\omega\in\Omega$ has a representation $\omega=u+\hat{w}$, with $\hat{w}\in\hat{V}$ (see Fig. \ref{poly}). In this representation the outcome probability of the effect $e$ on the state $\omega$ reads  $p(e|w)= \lambda p(\tilde{e}|\omega)=\lambda(u + \hat{e}).(u + \hat{\omega})=\lambda(1+\hat{e}.\hat{\omega})$  
\cite{Self3}, where dot represents euclidean inner-product in $\mathbb{R}^n$. Set of the vectors $\hat{\omega}$ corresponding to the states $\omega\in\Omega$ forms a convex-compact set $\hat{W}_{s}\subset\hat{V}$. For the n-gonal case, norm of these vectors satisfy the bound $||\hat{\omega}||_{2}\le r_{n}$ with exactly $n$ vectors saturating the bound. Similarly the vectors $\hat{e}$ forms another convex-compact set $\hat{W}_{e}\subset\hat{V}$ and $||\hat{e}||_{2}\le r_{n}$ with exactly $n$ vectors saturating the bound. Self-duality for the odd-gonal cases imply $\hat{W}_{s}=\hat{W}_{e}$ which is not the case for even $n$. 
\begin{theorem}\label{theorem1}
	GPTs with state space $\Omega_n$ with $n>3$ (for odd $n$) and with $n>4$ (for even $n$) are not compatible with IS.
\end{theorem}
\begin{proof}
	Let us first consider an n-gonal state space with odd $n$ and $n>3$.
	Without loss of generality, consider the two neighboring states $\omega_0, \omega_1\in \Omega_n$. According to Remark \ref{remark1} the measurement that optimally discriminate these states  consists of one of the effects corresponding to the vectors $\hat{e}_{n-k}\in\hat{W}_{e}$ such that $||\hat{e}_{n-k}||_{2}=r_{n}$, with $k\in\{0,\cdots,n-1\}$ and its complementary effects. With such a measurement the error reads as,
	\begin{equation} \label{E_min}
	p_E=\frac{1}{2}\left[1+\frac{1}{1+r_{n}^{2}} \hat{e}_{n-k}.(\hat{\omega}_{0}-\hat{\omega}_{1})\right].
	\end{equation}
	Let us denote the angle between $\hat{e}_{n-k}$ and $(\hat{\omega}_{0}-\hat{\omega}_{1})$ as $\theta_k$. It is evident from \eqref{E_min} that for minimal error $k$ should be chosen in a way that $|\theta_k-\pi|\to0$. However, the self-duality of odd-gonal theory demands that $\theta_k=\frac{\pi}{2}+(2k+1)\frac{\pi}{n}$. Then, a straightforward calculation shows that minimal error discrimination is achieved for $k=[\frac{n}{4}]$. For this optimal measurement, probability of clicking $e_{n-k}$ when the input state is $\omega_0$ is given by $p=\frac{1}{1+r_{n}^2}[1+r_{n}^2\cos\{\frac{2\pi}{n}(k+1)\}]$ and probability of clicking $\bar{e}_{n-k}$ on $\omega_1$ is given by $\bar{p}=\frac{r_{n}^2}{1+r_{n}^2}[1-\cos(\frac{2\pi}{n}k)]$. An elementary trigonometric argument ensures that the probabilities $p$ and $\bar{p}$ are not same for any $\Omega_n$, with odd $n$ and $n\ge3$. The absolute difference of these two probabilities, however, decreases with increasing $n
		$ (see Fig.\ref{figodd}). 
	
	The proof for the even-gonal case is similar to the odd-gonal case. We provide the detailed proof in the Supplemental Materials \cite{Supple}.  
\end{proof}

\begin{figure}[t]
	\includegraphics[scale=0.65]{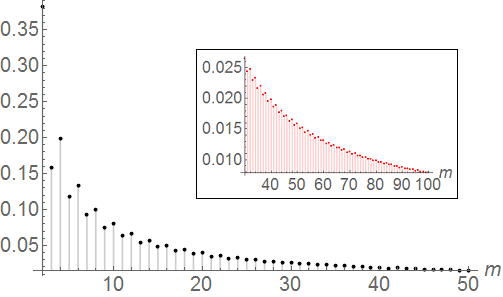}
	\caption{(Color on-line) Absolute difference between $p$ and $\bar{p}$ is plotted against $m\ge2, m\in \mathbb{Z}$, where $\Omega_{(2m+1)}$ is the corresponding odd-gon state space. Inset depicts magnified plot for higher values of $m$.}\label{figodd}
\end{figure}
We have shown that all polygonal state spaces $\Omega_n$, with $n\ge5$ are incompatible with IS. Now the question arises as to what happens for $n=4$ which corresponds to the marginal state space of the most general two-input-two output bipartite NS correlations. This state apace is also known by the name {\it squit} whose center corresponds to the marginal state of the famous Popescu-Rohrlich correlation \cite{Popescu'94}. Here we have the following observation about squit state space.
\begin{observation}\label{lemma2}
	Any pair of pure states in squit can be discriminated perfectly. 
\end{observation}
It is possible to generalize IS that applies to the ensemble of mixed states. A GPT is said to be compatible with generalized information symmetry (GIS) if every pair of states each having identical {\it minimal} type subjective ignorance can be optimally discriminated with symmetric error measurement. While a pure state is the state of maximal knowledge, {\it i.e.}, contains no subjective ignorance, a state $\omega$ is said to have {\it minimal} type subjective ignorance if it allows a convex decomposition in terms of two distinguishable pure states, {\it i.e.}, $\omega=p\omega_i+(1-p)\omega_j$ for some perfectly distinguishable pair of pure states $\omega_i$ and $\omega_j$. Two such states $\omega=p\omega_l+(1-p)\omega_j$ and $\omega^\prime=q\omega_k+(1-q)\omega_l$ are said to have identical subjective ignorance when $p=q$. It turns out that {\it squit} state space does not satisfy GIS while quantum theory is perfectly compatible with GIS \cite{Supple}.

\emph{Discussions.--}
The newly identified symmetric primitive, namely the Information Symmetry, has important implications in the axiomatic derivation of Hilbert space quantum mechanics as it puts nontrivial restrictions  on the state space structure of generalized probabilistic models. While the state space of quantum theory is perfectly compatible with IS, we find that the polygonal state spaces do not satisfy this elementary symmetry condition, or its generalized version.

In this context it is worth mentioning a couple of other features of the structure of GPTs, which though interesting, are not powerful enough to exclude various categories of models
while allowing for quantum and classical mechanics in the manner of IS. First, the notion of logical bit-symmetry \cite{Muller12} imparts self-duality on the state space leading to the exclusion of even-gonal state spaces only, but not the odd-gonal ones  \cite{Janotta'11}. Secondly, polygonal 
state spaces lack  well defined purification for all states \cite{Winczewski18}. However,  the state space of the classical bit also lacks this particular property, whereas it satisfies IS. On the other extreme,  the 'toy bit' model of  Spekkens \cite{Spekkens07} does not satisfy IS though it may allow well defined purification \cite{Supple}. 

To summarize, IS imparts a remarkable restriction on the state space structure, excluding all regular polygonal state spaces as well as the Spekkens  model, thus representing a more stringent structural constraint compared to self-duality. Moreover, unlike bit-symmetry, IS assumes no constraint on the dynamics of the theory. Before concluding, note that though it can be shown that
the state space of the bipartite NS box with a Bell measurement is equivalent to a Bloch ball \cite{Czekaj18}, the formulation of IS is more general and does not involve any structure from composite systems. Finally, it may be interesting to explore implications of IS on other state space structures as well as generalizations of IS for ensembles prepared with bias.

\begin{acknowledgments}
	We would like to acknowledge stimulating discussions with Guruprasad Kar and Ashutosh Rai. SSB acknowledges fruitful discussions with Giulio Chiribella and Michele Dall'Arno. MB likes to acknowledge discussions with Karol Horodecki during QIPA-18 at HRI, Allahabad, India. SA acknowledges the support through Research Grant of INSPIRE Faculty Award of MB which supported her visit at S. N. Bose National Center for Basic Sciences. SSB acknowledges his visit at S. N. Bose National Center for Basic Sciences. MB acknowledges support through an INSPIRE-faculty position at S. N. Bose National Center for Basic Sciences by the Department of Science and Technology, Government of India. SSB is supported by the John Templeton Foundation through grant 60609, Quantum Causal Structures. The opinions expressed in this publication are those of the authors and do not necessarily reflect the views of the John Templeton Foundation. ASM acknowledges support from the DST project DST/ICPS/QuEST/2019/
	Q98.
\end{acknowledgments}

\onecolumngrid
\section{Supplementary}
\subsection{Elements of Convex Geometry}
Here we recall some definitions from convex geometry \cite{Book1,Book2} that are relevant for our purpose.

{\bf Definition} [{\it Convex set}]. A set $\mathcal{S} \subset \mathbb{R}^{n}$ is said to be convex {\it if and only if} $\lambda x_{1}+(1-\lambda) x_{2}\in \mathcal{S},~~\forall x_{1}, x_{2}\in \mathcal{S}$ where $\lambda\in[0,1]$.

The points residing on the boundary of a closed convex set $\mathcal{S}$, which can not be written as a strict convex combination of other two distinct points in $\mathcal{S}$ are said to be {\it extreme points} of the set $\mathcal{S}$. More precisely, 

{\bf Definition} [{\it Extreme point}]. A point $x\in \mathcal{S}$ is said to be extreme point {\it if and only if} $x=\lambda x_{1} +(1-\lambda) x_{2} \implies x_{1}=x_{2}=x$, where $\lambda\in(0,1)$ and $x_{1},x_{2}\in \mathcal{S}$.

A convex set is called convex cone if it satisfies a further condition.

{\bf Definition} [{\it Convex cone}]. A convex set $\mathcal{C}$ is said to be a convex cone {\it if and only if} $\forall x\in \mathcal{C},~~\lambda x \in \mathcal{C}$ for every $\lambda \geq 0$.

An important notion in convex geometry is the {\it Face} of a convex set. Face can be defined in two different ways: -- geometrically and algebraically. These two definitions exhibit potential difference in case of non-polyhedral convex sets.

Let us use the following notations:\\
$\mathcal{M}^n\Rightarrow$ set of symmetric $n \times n$ matrices;\\
$\mathcal{M}^n_+\Rightarrow$ set of symmetric positive semi-definite matrices, {\it i.e.}, $\mathcal{M}^n_+:=\{M\in\mathcal{S}^n~|~M\ge0\}$;\\
$\mathcal{M}^n_{++}\Rightarrow$ set of symmetric positive definite matrices, {\it i.e.}, $\mathcal{M}^n_{++}:=\{M\in\mathcal{S}^n~|~M>0\}$.\\
Both $\mathcal{M}^n$ and $\mathcal{M}^n_+$ are convex cones, and the later is called positive semi-definite convex cone. $\mathcal{M}^n_{++}$
comprises cone interior of $\mathcal{M}^n_+$.

{\bf Definition} [{\it Face -- Geometric}]. A subset $\mathcal{F}$ is called a face of the convex set $\mathcal{P}\subseteq\mathbb{R}^{n}$ if there exists a supporting hyperplane of  $\mathcal{P}$ as, $\langle v, *\rangle=d$, such that,  $\mathcal{F}=\mathcal{P}\cap\{x| \langle v, x \rangle=d\}$. Evidently, $\phi$ and $\mathcal{P}$ by itself are the trivial faces of $\mathcal{P}$.

{\bf Definition} [{\it Face -- Algebraic}]. A face $\mathcal{F}$ of a convex set $\mathcal{P}\subseteq\mathbb{R}^{n}$ is a closed convex subset of $\mathcal{P}$ such that for any $x\in \mathcal{F}$ and any line segment $[a,b]\subset\mathcal{P}$ with $x\in(a,b)$ implies $a,b\in \mathcal{F}$.

The geometric definition of face gives the idea of an {\it exposed face}. However, the algebraic definition of face captures both {\it exposed} and {\it non-exposed faces}. In the Fig.\ref{Fig3} a sharp distinction between the above two faces are explained.\\
\begin{figure}[h!]
	\includegraphics[scale=0.25]{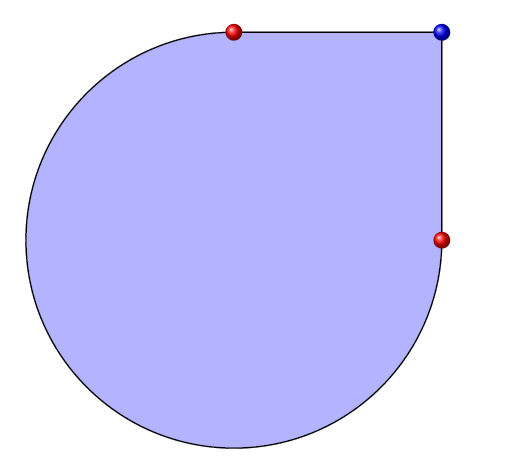}
	\caption{(Color on-line) On the non-polyhedral set two red points denote the extreme but non-exposed faces, however the blue point represents the exposed face. It is also evident that although the blue point satisfies the geometrical criterion of being a face, the red points do not. However, all of them are accepted according to the algebraic definition of face.}\label{Fig3}
\end{figure}
Another important notion in convex geometry is the {\it ray} of a convex cone. For any $x$ in the convex cone $\mathcal{C}$, $\lambda x$ is said to be a ray generated by $x$ for $\lambda\geq0$.

{\bf Definition} [{\it Extreme ray}]. The ray generated by a non-null point $x\in \mathcal{C}$ is said to be an extreme ray for $\mathcal{C}$ if the ray $\{\lambda x|\lambda\geq0\}$ is a face of $\mathcal{C}$.

Another refined notion  about ray of a cone is {\it exposed ray}. 

{\bf Definition} [{\it Exposed ray}]. The extreme ray generated by a non-null point $x\in \mathcal{C}$ is said to be an exposed ray for $\mathcal{C}$ if the ray is an exposed face of $\mathcal{C}$.

In the structure of general probabilistic theories, the ray joining the {\it null effect} and any of the pure states, is an exposed ray for the state cone. A non-exposed extreme ray can only exist for a non-polyhedral convex cone.

\subsection{Argument supporting Remark 1}
Consider two arbitrary states $\omega_1,\omega_2\in\Omega$. Suppose that the observable $M\equiv\{e_1,e_2~|~e_1+e_2=u\}$ discriminates the pair of states is a way that the state $\omega_i$ is guessed while the effect $e_i$ clicks. Probability of error in guessing is therefore,
\begin{equation}\label{a1}
p_E=\frac{1}{2}\left[p_{12}+p_{21}\right] =\frac{1}{2}\left[p(e_1|\omega_2)+p(e_2|\omega_1)\right] .
\end{equation}
Effects $e_i$'s are in general mixed and allow convex decomposition in terms of pure effects, {\it i.e.}, $e_i=\sum_{k}p_{k}e^{k}_i$, where all $e^{k}_i$'s are pure effects and $p_k\ge0~\forall~k~\&~\sum_kp_k=1$. Therefore, we can write $p(e_1|\omega_2)=\sum_{k}p_{k}p(e^{k}_1|\omega_2)$ and $p(e_2|\omega_1)=p(u-e_1|\omega_1)=1-\sum_{k}p_{k}p(e^{k}_1|\omega_2)$, and consequently Eq.(\ref{a1}) becomes,
\begin{equation}\label{a2}
p_E=\frac{1}{2}\left[1-\sum_kp_kp(e^k_1|\omega_2-\omega_1)\right] .
\end{equation}
From this expression, it is clear that minimum error occurs when $\sum_kp_kp(e^k_1|\omega_2-\omega_1)$ gets maximized which will be obtained for some pure effect.
\subsection{Quantum theory and IS}
\begin{figure}[b]
	\includegraphics[scale=0.3]{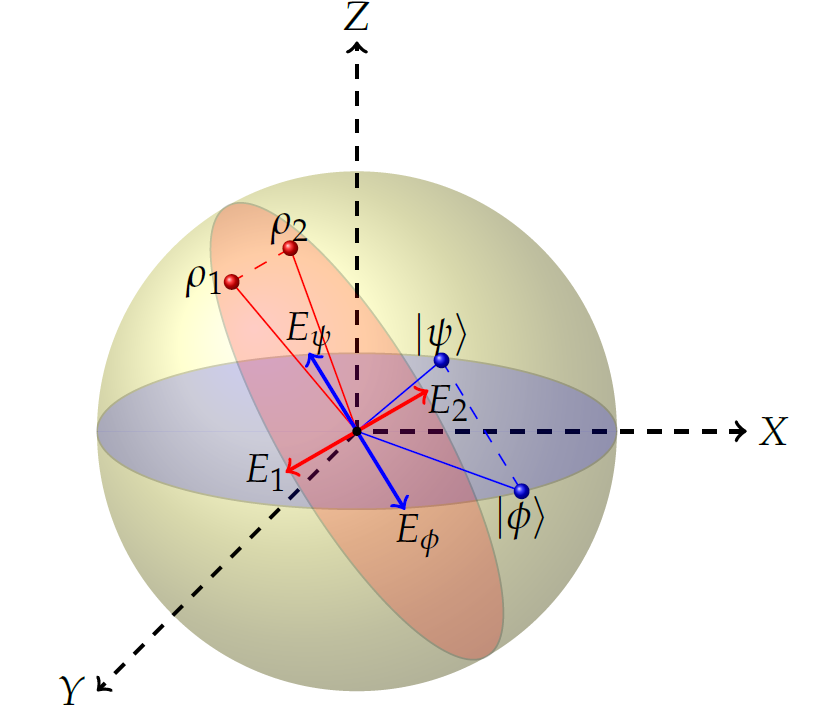}
	\caption{(Color on-line) Bloch sphere: state space of a two-level quantum system (qubit). For randomly prepared pure states $\ket{\psi}$ and $\ket{\phi}$ MESD is obtained through Helstrom measurement $M=\{E_{\psi},E_{\phi}\}$. The angle between the vectors corresponding to the state $\ket{\psi}$ and projector $E_{\psi}$ is equal to the angle between the vectors corresponding to the state $\ket{\phi}$ and projector $E_{\phi}$. Vectors denoted by red color indicate a similar fact for two equidistant mixed states $\rho_1$ and $\rho_2$.}\label{bs}
\end{figure}
In minimum error state discrimination the aim is to guess one of the states from a given ensemble with minimum error. In the simplest scenario consider that two states $\psi=\ket{\psi}\bra{\psi}$ and $\phi=\ket{\phi}\bra{\phi}$ are prepared with probability $p_{\psi}$ and $p_{\phi}=1-p_{\psi}$. The most general strategy to guess the state perfectly is to perform a measurement $M\equiv\{E_{\psi},E_{\phi}~|~E_{\psi},E_{\phi}\ge 0~\&~E_{\psi}+E_{\phi}=\mathbb{I}\}$ such that the outcome corresponding to $E_{\psi}~(E_{\phi})$ is taken to indicate that the state was $\psi~(\phi)$. The probability of error in determining the state is therefore,
\begin{eqnarray}
p_{er}&=&p_{\psi}\tr(\psi E_{\phi})+p_{\phi}\tr(\phi E_{\psi}),\nonumber\\
&=&p_{\psi}-\tr[(p_{\psi}\psi-p_{\phi}\phi)E_{\psi}].
\end{eqnarray}
Clearly, minimum error will be achieved when $E_{\psi}$ is a projector onto the positive eigenket of the operator $p_{\psi}\psi-p_{\phi}\phi$. The state $\psi$ and $\phi$ span a two-dimensional subspace. Without any loss of generality we can choose an orthogonal basis $\{\ket{0},\ket{1}\}$ and express the states $\psi,\phi$ as
\begin{subequations}
	\begin{align}
	\ket{\psi}&=\cos\theta\ket{0}+\sin\theta\ket{1},\\
	\ket{\phi}&=\cos\theta\ket{0}-\sin\theta\ket{1}.
	\end{align}
\end{subequations}
Accordingly, the eigenvalues of the operator $p_{\psi}\psi-p_{\phi}\phi$ become,
\begin{equation}
a_{\pm}=\frac{1}{2}\left(p_{\psi}-p_{\phi}\pm\sqrt{1-4p_{\psi}p_{\phi}\cos^22\theta} \right), 
\end{equation}
and consequently, the minimal error is given by the so-called Helstrom quantity,
\begin{equation}
p^{\min}_{er}=\frac{1}{2}\left(1-\sqrt{1-4p_{\psi}p_{\phi}|\braket{\psi|\phi}|^2} \right).
\end{equation}
For the particular case $p_{\psi}=p_{\phi}=\frac{1}{2}$, we have $p^{\min}_{er}=\frac{1}{2}\left(1-\sqrt{1-|\braket{\psi|\phi}|^2} \right)=\frac{1}{2}(1-\sin2\theta)$. The optimal measurement is symmetrically located about the input states, {\it i.e.}, $E_{\psi}$ and $E_{\phi}$ are projectors on $\frac{1}{\sqrt{2}}(\ket{0}+\ket{1})$ and $\frac{1}{\sqrt{2}}(\ket{0}-\ket{1})$,  respectively.  A straightforward calculation gives, $\tr({\psi E_{\phi}})=\frac{1}{2}(1-\sin2\theta)=\tr({\phi E_{\psi}})$ establishing compatibility of quantum theory with IS (see Fig. \ref{bs}).

\subsection{Proof of Theorem 1 for even n}
The proof demands to establish that every $\Omega_n$ violates IS for every even $n>4$. As already mentioned in the main text, depending upon the relative positioning of the symmetrically  distinguishable pair of pure states the even-gonal state spaces $\Omega_n$'s are characterized in two classes,
\begin{itemize}
	\item[(I)] $n=4m$, with $m\in\{2,3,\cdots\}$, and
	\item[(II)] $n=4m+2$, with $m\in\{1,3,\cdots\}$.
\end{itemize}  
In the following we treat these two cases separately. 

{\bf Case- I:} In this case we  show that no two even-ordered neighboring pure states can be optimally distinguished with symmetric error measurement. Without loss of any generality we consider one of the states $\omega_{0}$. Then  other even ordered neighboring state relative to this is $\omega_{2l}$. For $n=4m$-type polygon state space, the state $\omega_{2m}$ resides exactly opposite to $\omega_{0}$. The state $\omega_{2m}$ and its immediate predecessor $\omega_{2m-1}$ and successor  $\omega_{2m+1}$ are perfectly distinguishable from $\omega_{0}$. The symmetry of even-gon state space further implies that every pair of neighboring states on either side of $\omega_{0}$ have exactly 
the same status in respect of discrimination, {\it i.e.}, while discriminating $\omega_{0}$ either from $\omega_{k}$ or from $\omega_{n-k}$ the errors are same. So, in the following we consider only the pairs $(\omega_{0},\omega_{2l})$ with $l\in\{1,...,(m-1)\}$.
\begin{figure}[b]
	\includegraphics[scale=0.35]{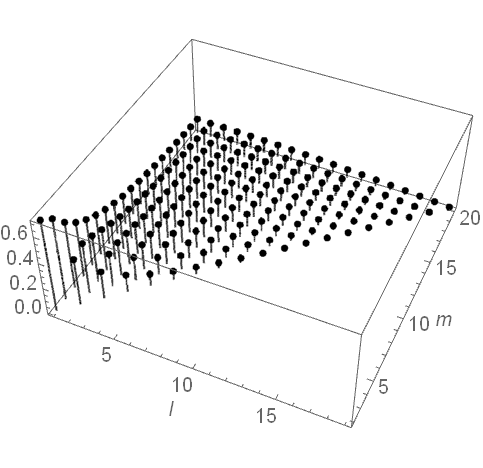}
	\caption{Absolute difference between the two error probabilities in Eqs.(\ref{4m}) is plotted against $m\ge2,~ m\in \mathbb{Z}$ and $l\le(m-1),~ l\in \mathbb{Z}$, where $4m$ is the corresponding even-gon structure and $l$ as mentioned in the text.}\label{fig2}
\end{figure}

The angle between the vectors $\hat{\omega}_{0}$ and $\hat{\omega}_{2l}$ is given by
\begin{equation*}
\theta_{l}=2l\frac{2\pi}{n},
\end{equation*}
and the angle between $\hat{e}_{n-k}$ and $(\hat{\omega}_{2l}-\hat{\omega}_{0})$ is 
\begin{equation*}
\theta_{k}= \frac{\pi}{2}+\frac{\theta_{l}}{2}+(2k+1)\frac{\pi}{n}.
\end{equation*}
To optimize the total error in this discrimination task, $k$ should be so chosen that $|\theta_{k}-\pi|\to 0$. Hence, $k$ will be the closest integer to $(m-l+\frac{1}{2})$, i.e., either $k=(m-l)$ or $k=(m-l+1)$. Straightforward calculations lead to,
\begin{subequations}\label{4m}
	\begin{align}
	p(e_{n-k}|\omega_{2l})&=\frac{1}{2}\left( 1+r_{n}^{2}\cos\left[(4l+2k+1)\frac{\pi}{n}\right] \right),\\
	p(\bar{e}_{n-k}|\omega_{0})&=\frac{1}{2}\left(1-r_{n}^{2}\cos\left[ (2k+1)\frac{\pi}{n}\right]\right).
	\end{align}
\end{subequations}
For symmetric error we require $p(e_{n-k}|\omega_{2l})=p(\bar{e}_{n-k}|\omega_{0})$, which further implies $\cos\left[ \frac{l\pi}{m}+(2k+1)\frac{\pi}{n}\right] =-\cos\left[ (2k+1)\frac{\pi}{n}\right] $, i.e., $\frac{l\pi}{m}$ is an odd multiple of $\pi$, which is not possible since $l\in\{1,...,(m-1)\}$. (see Fig. \ref{fig2}).
\begin{figure}[t]
	\includegraphics[scale=0.45]{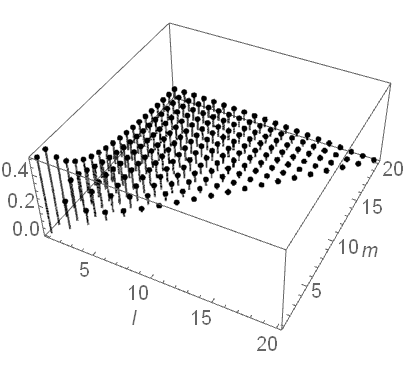}
	\caption{Absolute difference between the two error probabilities in Eqs.(\ref{4m+2}) is plotted against $m\ge2,~ m\in \mathbb{Z}$ and $l\le m,~ l\in \mathbb{Z}$, where $4m+2$ is the corresponding even-gon structure and $l$ as mentioned in the text.}\label{fig3}
\end{figure}

{\bf Case- II:} Without loss of generality, in this case it is sufficient to choose the pairs of states $\omega_{0}$ and $\omega_{2l-1}$ with $l\in\{1,2,...,m\}$.
The angle between vectors $\hat{\omega}_{0}$ and $\hat{\omega}_{2l-1}$ is given by
\begin{equation*}
\theta_{l}=(2l-1)\frac{2\pi}{n},
\end{equation*}
and the angle between $\hat{e}_{n-k}$ and the vector $(\hat{\omega}_{2l-1}-\hat{\omega}_{0})$ is
\begin{equation*}
\theta_{k}=\frac{\pi}{2}+\frac{\theta_{l}}{2}+(2k+1)\frac{\pi}{n}.
\end{equation*} 
Further calculation gives,
\begin{subequations}\label{4m+2}
	\begin{align}
	p(e_{n-k}|\omega_{2l-1})&=\frac{1}{2}\left( 1+r_{n}^{2}\cos\left[ \theta_{l}+(2k+1)\frac{\pi}{n}\right] \right),\\
	p(\bar{e}_{n-k}|\omega_{0})&=\frac{1}{2}\left( 1-r_{n}^{2}\cos\left[ (2k-1)\frac{\pi}{n}\right] \right).
	\end{align}
\end{subequations}
For symmetric errors we require $\cos\left[ (4l+2k-1)\frac{\pi}{n}\right] =\cos\left[ (2k-1)\frac{\pi}{n}\right] $, i.e., $\frac{4l\pi}{n}=\frac{2l\pi}{2m+1}$ is an odd multiple of $\pi$, which is again not possible. However, an interested reader can see Fig. \ref{fig3} to get an idea that how the asymmetric error difference is changing.

\subsection{Argument supporting Observation 1}
The nontrivial supporting hyperplane corresponding to any pure effect can be visualized as a first order face of the corresponding state space.
Any pair of pure states in squit theory always lie on a pair of parallel nontrivial supporting hyperplanes of the effect space. Therefore, these two states can be perfectly discriminated by a measurement consisting of two pure effects on these two parallel hyperplanes.

\subsection{GIS: Squit state space and quantum theory}
\begin{center}
\begin{figure}[h!]
	\includegraphics[scale=0.3]{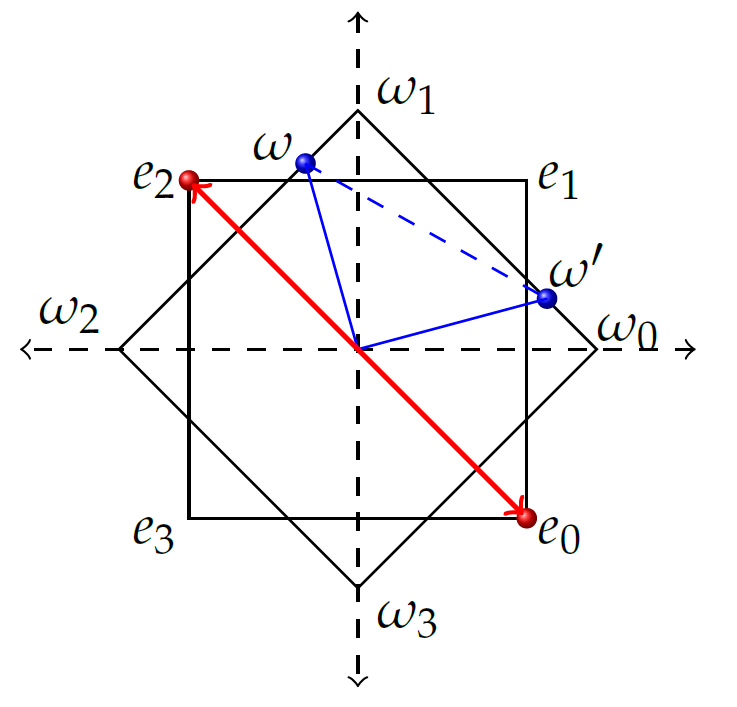}
		\caption{(Color on-line) Projection of the state and effect cones of the squit model on the normalized state plane.}\label{squit}
	\end{figure}
\end{center}
{\bf Squit state space:} Here our aim is to show that the squit theory is not compatible with GIS. Consider the symmetric MESD between two mixed states $\omega$ and $\omega^{\prime}$ as shown in Fig \ref{squit}. Note that the states are of minimal ignorance type and both have same subjective ignorance. We  now show that though a pair of mixed effects can discriminate the states with symmetrically distributed error, but the total error in that case is not minimum; rather the minimum error is obtained for pure effects with error corresponding to one effect being zero, and full for the other one.

Consider the decomposition of $\omega$ and $\omega^{\prime}$ in terms of pure states:
\begin{subequations}
	\begin{align}
	\omega&=p\omega_1+(1-p)\omega_2,\\
	\omega^{\prime}&=p\omega_0+(1-p)\omega_1.
	\end{align}
\end{subequations}
To discriminate these states, consider a measurement $M\equiv\{e,e^{\prime}|~e=re_2+(1-r)e_3~\&~e^{\prime}=u-e\}$ to guess the state as $\omega~(\omega^{\prime})$ when the effect $e~(e^{\prime})$ clicks. A straightforward calculation shows  $p(e^{\prime}|\omega)= p(1-r)$ and $p(e|\omega^{\prime})=r(1-p)$. For $p=r$, both the errors are identical and accordingly, the total error turns out to be $p^{(s)}_{er}= 2p(1-p)$.

Consider now a different measurement consisting of two pure effects $\{e_2,e_4\}$. If the given state is $\omega$, then the effect $e_2$ clicks certainly, and $e_4$ never clicks, whereas, for $\omega'$ both the effects click. The error turns out to be $p_{er}=(1-p)$. Clearly, for $p>1/2$ we have $p_{er}<p^{(s)}_{er}$ establishing that the squit theory does not satisfy GIS.

{\bf Quantum theory:} GIS applies to pairs of states each having identical {\it minimal} type subjective ignorance. A quantum state having minimal type subjective ignorance is nothing but a density operator of rank 2. While two such states are considered then the analysis effectively boils down to the Bloch sphere and two such states having identical subjective ignorance lies in same distance from the center of the Bloch sphere  (Fig. \ref{bs}). The geometry of the Bloch sphere thus assures that GIS is satisfied in quantum theory.

\subsection{IS and the Spekkens toy-bit model}
Spekkens {\it toy-bit} theory was constructed in support of an epistemic view of quantum states \cite{Spekkens07}. This theory encompasses a wide variety of quantum phenomena along with the existence of nonorthogonal states that are impossible to discriminate perfectly. This model is based on a principle, called {\it knowledge balance principle} (KBP) which states that -- {\it ``If one has maximal knowledge, then for every system, at every time, the amount of knowledge one possesses about the ontic state of the system at that time must equal the amount of knowledge one lacks".} 

The most elementary system consists of four ontic states denoted as $`1'$, $`2'$, $`3'$, and $`4'$. This elementary system has only six epistemic states of maximal knowledge that are compatible with KBP:
\begin{align}
\left\{\!\begin{aligned}
1\vee 2,~~~~3\vee 4,~~~~1\vee 3,\\
2\vee 4,~~~~2\vee 3,~~~~1\vee 4~
\end{aligned}\right\}.	
\end{align}
Here the symbol $`\vee'$ denotes disjunction and it reads as `or'. For this elementary system, the only epistemic state with non-maximal knowledge is given by,
\begin{equation}
1\vee 2\vee 3\vee 4.
\end{equation}
The epistemic states of the toy theory can be treated as
uniform probability distributions over the ontic states. For instance, while probability distribution $(1/2, 1/2, 0, 0)$ is associated with the state $1\vee 2$, the state $1\vee 2\vee 3\vee 4$ corresponds to the distribution $(1/4, 1/4, 1/4, 1/4)$. Note that this particular toy model is not a GPT in true sense as it does not allow all possible convex mixtures of pure states as valid states. While the states $1\vee 2$ and $3\vee 4$ are mutually orthogonal as they have non-overlapping probability distributions over the ontic states, $1\vee 2$ and $1\vee 3$ correspond to non orthogonal states.

KBP also imposes restrictions on possible implementable measurements in this toy theory. The fewest ontic states that can be associated with a single outcome of a measurement is two. Thus, the
only valid reproducible measurements are those which partition the four ontic states into two sets of two ontic states. Therefore we have three possible measurements
\begin{align}
\left\{\!\begin{aligned}
M_1:=\left\lbrace 1\vee 2,~~3\vee 4 \right\rbrace ,\\
M_2:=\left\lbrace 1\vee 3,~~2\vee 4 \right\rbrace ,\\
M_3:=\left\lbrace 1\vee 4,~~2\vee 3 \right\rbrace ~
\end{aligned}\right\}.	
\end{align}
To discriminate the pair of nonorthogonal states $1\vee 2$ and $1\vee 3$, while the measurements $M_3$ is no good for minimal error discrimination, the measurements $M_1$ and $M_2$ do not provide symmetric errors. This establishes that toy-bit theory is incompatible with the principle of information symmetry.


\begin{thebibliography}{99}
\bibitem{Self2} In quantum mechanics this happens when dimension $d$ of the associated Hilbert space $\mathcal{H}_d$ is larger than two. For every pure state $\psi\equiv\ket{\psi}\bra{\psi}$ the effect $\left( \mathbb{I} - \psi\right) $ is extremal in the set of proper effects but not a ray extremal effect; $\ket{\psi}\in \mathcal{H}_d$, and $\mathbb{I}$ is the identity operator on $\mathcal{H}_d$.
\bibitem{Self3} This expression has close similarity with the corresponding qubit expression. If a two outcome positive-operator-valued-measurement $M\equiv\{ E:=\frac{\mu_1}{2}(\mathbb{I}+\mu_2\vec{a}.\vec{\sigma}),~\mathbb{I}-E~|~0\le\mu_1\le2,~~0\le\mu_2\le\mbox{min}\{\mu_1,2-\mu_1\}\}$ is performed on a qubit state $\rho=\frac{1}{2}(\mathbb{I}+\vec{n}.\vec{\sigma})\in\mathcal{D}(\mathbb{C}^2)$, probability of clicking the effect $E$ is given by $p(E|\rho)=\Tr(\rho E)=\frac{1}{2}(\mu_1+\mu_2\vec{a}.\vec{n})$; here $\vec{a}.\vec{\sigma}\equiv a_x\sigma_x+a_y\sigma_y+a_z\sigma_z$; $\vec{a},\vec{n}\in\mathbb{R}^3$ with $|\vec{n}|\le1$ and $\sigma$'s are the Pauli operators.
\end{thebibliography}
\end{document}